\documentclass[letterpaper, 10 pt, conference]{ieeeconf}  

\IEEEoverridecommandlockouts                              
\usepackage{cite}

\usepackage{amsmath,amssymb,amsfonts}
\usepackage{amsthm}
\usepackage{algorithm2e}
\usepackage{graphicx}
\usepackage{makecell}
\usepackage{caption}
\usepackage{textcomp}
\usepackage{comment}
\usepackage{theoremref}

\DeclareMathOperator*{\argmax}{arg\,max}

\newcommand{\rom}[1]{\uppercase\expandafter{\romannumeral #1\relax}}

\theoremstyle{remark}
\newtheorem{remark}{Remark}

\theoremstyle{definition}
\newtheorem{definition}{Definition}
\newtheorem{theorem}{Theorem}
\newtheorem{lemma}{Lemma}

\title{\LARGE \bf
On Bounds for Greedy Schemes in String Optimization based on Greedy Curvatures
}

\author{Bowen Li, Brandon Van Over, Edwin K. P. Chong, and Ali Pezeshki
\thanks{This work is supported in part by the AFOSR under award FA8750-20-
2-0504 and by the NSF under award CNS-2229469.}
\thanks{B. Li, B. Van Over, E. K. P. Chong and A. Pezeshki are with the Department of Electrical and Computer Engineering, Colorado State University, Fort Collins, CO 80523, USA. 
        {\tt\small \{bowen.li, b.van\_over, edwin.chong, ali.pezeshki\}@colostate.edu } }%
}

\begin{document}

\maketitle
\thispagestyle{empty}
\pagestyle{empty}

\begin{abstract}

We consider the celebrated bound introduced by Conforti and Cornu\'{e}jols (1984) for greedy schemes in submodular optimization. The bound assumes a submodular function defined on a collection of sets forming a matroid and is based on greedy curvature. We show that the bound holds for a very general class of string problems that includes maximizing submodular functions over set matroids as a special case. We also derive a bound that is computable in the sense that they depend only on quantities along the greedy trajectory. We prove that our bound is superior to the greedy curvature bound of Conforti and Cornu\'{e}jols. In addition, our bound holds under a condition that is weaker than submodularity. 

\end{abstract}

\section{INTRODUCTION}
In many sequential decision-making or machine learning problems, we encounter the problem of optimally choosing a \emph{string} (ordered set) of actions over a finite horizon 
to maximize a given objective function. String optimization problems have the added complexity relative to \emph{set} optimization problems in that the objective function depends on both the set of actions taken and the order in which the actions are taken. In such problems, determining the optimal solution (optimal string of actions) can become computationally intractable with increasing size of state/action space and optimization horizon. Therefore, we often have to resort to approximate solutions. One of the most common approximation schemes is the \emph{greedy scheme}, in which we sequentially select the action that maximizes the increment in the objective function at each step. A natural question that arises is how good is the greedy solution to a \emph{string optimization} problem relative to the optimal solution? 


In this paper, we derive a \emph{ratio bound} for the performance of greedy solutions to string optimization problems relative to that of the optional solution. By a ratio bound, we mean a bound of the form
\[
\frac{f(G_K)}{f(O_K)}\geq \beta,
\]
where $f(G_K)$ and $f(O_K)$ are the objective function values of the greedy solution and optimal solution, respectively, and $K$ is the horizon. The bound guarantees that the greedy solution achieves at least a factor of $\beta$ of the optimal solution. Our focus is on establishing a bound in which the factor $\beta$ is easily computable.  

\noindent\textbf{Prior work:} Bounding greedy solutions to set optimization problems has a rich history, especially in the context of \emph{submodular} set functions. The most celebrated results are from Fisher et al.\  \cite{fisher1978analysis} and Nemhauser et al.\ \cite{nemhauser1978analysis}. They showed that $\beta \geq 1/2$ over a finite general set matroid \cite{fisher1978analysis} and $\beta \geq 1- ((K-1)/K)^{K} > (1-e^{-1})$ over a finite uniform set matroid with $K$ being the horizon \cite{nemhauser1978analysis}. Improved bounds (with $\beta$ larger than $1-e^{-1}$) have been developed by Conforti and Cornu\'{e}jols \cite{conforti1984submodular}, Vondrak \cite{vondrak2010submodularity, calinescu2011maximizing}, and Wang et al.\ \cite{wang2016approximation}, by introducing various notions of \emph{curvature} for set submodular functions under uniform and/or general matroid settings.

However, for most of these bounds, computing the value of $\beta$ is intractable and hence the bound is effectively not computable in problems with large action spaces and decision horizons. In addition, the value of $\beta$ depends on the values of the objective function beyond the optimization horizon. One notable exception is the bound derived by  Conforti and Cornu\'{e}jols \cite{conforti1984submodular} that is based on a quantity called \textit{greedy curvature}. This bound is easily computable. Another notable exception is the bound derived by Welikala et al.~\cite{welikala2022new}, which is also computable and in some problems has shown to be larger than other curvature bounds. However, the value of $\beta$ for this bound depends on values of the objective function beyond the optimization horizon.   



Recently, Zhang et al.\ \cite{zhang2015string} extended the concept of set submodularity to string submodularity and established a $\beta=(1-e^{-1})$  bound for greedy solutions of string optimization problems. They also derived improved bounds involving various notions of curvature for string submodular functions. However, these curvatures have the same computational intractability issue as mentioned above. Alaei et al.\ \cite{alaei2021maximizing} proved $(1-e^{-1})$ and $(1-1/e^{(1-1/e)})$ bounds for online advertising and query rewriting problems, both of which can be formulated under the framework of string submodular optimization.    

\noindent\textbf{Main contributions:} The contributions of this paper are as follows:
\begin{enumerate}
    \item We first derive a computable bound for greedy solutions to string optimization problems by extending the notion of greedy curvature and the proof technique used by Conforti and Cornu\'{e}jols \cite{conforti1984submodular} from set functions to string functions. Our bound coincides with the greedy curvature bound in \cite{conforti1984submodular} if reducing our string objective function to a set objective function. However, in deriving our bound, we don't use submodularity. Rather we rely on weaker conditions on the string objective function. Therefore, our bounding result is stronger. This result is reported in \thref{GenConforti}.   

    \item We then establish an even stronger result. We derive another computable bound that relies on weaker assumptions than those used in deriving our first bound. This is done in \thref{topK}. 
    
    \item We then show that the value $\beta_2$ of the bound in \thref{topK} is larger than the value $\beta_1$ of the bound in \thref{GenConforti} (i.e., the generalization of the  Conforti and Cornu\'{e}jols bound). This is shown in \thref{betterThm}.    
    
\end{enumerate}


\noindent\textbf{Organization:} The paper is organized as follows. Section \rom{2} introduces all the mathematical preliminaries and formulates the string optimization problem. Section \rom{3} presents the mains results regarding our computable bounds. Applications of our theoretical results to task scheduling and sensor coverage problems are demonstrated in Section \rom{4}. Finally, conclusions are given in Section \rom{5}.  

\section{DEFINITIONS FOR STRING OPTIMIZATION}

In this section, we introduce notation, terminology, and definitions that we use in our paper and present a formulation for a general string optimization problem. 

\begin{definition}
    Let $\mathbb{S}$ be a set. In our context, $\mathbb{S}$ is called the \textit{ground set} and its elements are called \textit{symbols}.
\end{definition}
\begin{enumerate}
    \item Let $s_1,s_2,\ldots,s_k \in \mathbb{S}$. Then $S:= s_{1}s_{2}\cdots s_{k}$ is a \textit{string} with length $|S| = k$. 
    \item Let $K$ be a positive integer, called the \textit{horizon}.
    \item Let $\mathbb{S}_{K}$ be the set of all positive strings of length up to $K$, including the empty string $\varnothing$. It is also called the \textit{uniform string matroid of rank} $K$.
\end{enumerate}

\begin{definition}
    Consider two strings $S = s_{1}s_{2}\cdots s_{m}$ and $T = t_{1}t_{2}\cdots t_{n}$ in $\mathbb{S}_{K}$. 
    \begin{enumerate}
        \item Define the \textit{concatenation} of $S$ and $T$ as $ST := s_{1}s_{2}\cdots s_{m}t_{1}t_{2}\cdots t_{n}$
        \item We say that $P$ is a \textit{prefix} of $S$ if $S = PU$ for some $U \in \mathbb{S}_{K}$, in which we write $P \preceq S$. 
        \item We say that $\mathbb{T} \subseteq \mathbb{S}_{K}$ is \textit{prefix-closed} if for all $S \in \mathbb{T}$ and $P \preceq S$, $P \in \mathbb{T}$.
        \item Let $f: \mathbb{S}_{K} \rightarrow \mathbb{R}$ be the objective function. 
    \end{enumerate}
\end{definition}

\begin{definition}
    Let $\mathbb{X} \subseteq \mathbb{S}_{K}$. Then $(\mathbb{X},\mathbb{S})$ is a \textit{finite rank $K$ string matroid} if 
    \begin{enumerate}
        \item $|A| \leq K$ for all $A \in \mathbb{X}$. 
        \item If $B \in \mathbb{X}$ and $A \preceq B$, then $A \in \mathbb{X}$.
        \item For every $A,B \in \mathbb{X}$ where $|A| + 1 = |B|$, there exists $b$ that is a symbol in $B$ and $Ab \in \mathbb{X}$. 
    \end{enumerate}
\end{definition}

\begin{definition}
    $f: \mathbb{S}_{K} \rightarrow \mathbb{R}$ is \textit{string submodular} on $(\mathbb{X},\mathbb{S})$ if 
    \begin{enumerate}
        \item $\forall A\preceq B \in \mathbb{X}$, $f(A) \leq f(B)$. 
        \item $\forall A\preceq B \in \mathbb{X} \text{ and } \forall j \in \mathbb{S}$ such that $Aj, Bj \in \mathbb{X}$,   $f(Aj)-f(A) \geq f(Bj)-f(B)$.
    \end{enumerate}
\end{definition}

\begin{remark} {\ }
\begin{enumerate}
    \item Unlike the permutation invariance in the set case, the order of the symbols in a string matters. Different orders of the same set of symbols represent different strings. 
    \item Note that $ST \in \mathbb{S}_{K}$ if and only if $m+n \leq K$. 
    \item In the matroid literature, a \textit{prefix-closed} collection is also said to satisfy the \textit{hereditary} property and is called as \textit{independence system}. An independence system does not require condition $(3)$ in Definition 3. 
\end{enumerate}    
    
\end{remark} 

\begin{remark}
The definitions of finite rank string matroid and string submodular function are introduced in \cite{van2023improved}. Their formulations are inspired by the definitions of set matroid and set submodular function in \cite{conforti1984submodular}. Most of the previous work is established in the set case. Here, we extend the theoretical results to the more general string case. 
\end{remark}

Consider an objective function $f$ and a prefix-closed $\mathbb{T}$ with $\max_{S \in \mathbb{T}} |S|= K$. The general string-optimization is given by
\begin{equation}
\label{string_opt}
\begin{aligned}
    & \text{maximize } f(S) \\
    & \text{subject to } S \in \mathbb{T}
\end{aligned}
\end{equation}

\begin{remark} {\ }
\begin{enumerate}
    \item The stipulation that $f(\varnothing) = 0$ is without loss of generality because if $f(\varnothing) \neq 0$, we can subtract $f(\varnothing)$ from all values of $f$ without changing the maximizer of the optimization problem (\ref{string_opt}).
    \item We deal with the case that $f(S) \geq 0$ for any $S \in \mathbb{T}$.
    \item The constraint set $\mathbb{T}$ here is not necessarily the uniform matroid $\mathbb{S}_{K}$ or a finite rank $K$ string matroid $(\mathbb{X},\mathbb{S})$. In particular, our analysis applies to general prefix-closed sets as constraint sets subject to certain assumptions. 
    \item The stipulation that $\max_{S\in \mathbb{T}}|S| = K$ is without loss of generality because we can always define $K:= \max_{S \in \mathbb{T}} |S|$.  
\end{enumerate}
\end{remark}

\begin{definition}
    Any solution to the optimization problem (\ref{string_opt}) is said to be \textit{optimal}, denoted by $O_{L} = o_{1}o_{2}\cdots o_{L}$, where $L \in \{1,\ldots, K\}$. 
\end{definition}

\begin{remark} {\ }
    \begin{enumerate}
        \item The length $L$ of an optimal solution $O_{L}$ could be anything from $1$ to $K$. 
        \item There may be multiple optimal solutions. 
        \item If $\mathbb{S}$ is finite, then $\mathbb{S}_{K}$ and $\mathbb{T}$ are finite, and an optimal solution always exists. 
    \end{enumerate}    
\end{remark}

\begin{definition}
    We define $G_{K} = g_{1}g_{2} \cdots g_{K}$ to be a \textit{greedy solution} if for all $k \in \{1,2,\ldots,K\}$, 
    
    \begin{equation*}
    \label{greedy_def}
    g_{k} = \argmax_{s \in \mathbb{S}: g_{1}\cdots g_{k-1}s \in \mathbb{T}} f(g_{1}\cdots g_{k-1}s). 
    \end{equation*}

\end{definition}
Note that we are going all the way to horizon $K$. This implicitly assumes that $f$ is nondecreasing along $G_{K}$ and $f(g_{1}) > 0$. Theses implicit conditions are reflected in assumption $\mathbf{A_{3}}$ in Section \rom{3}. 

\begin{remark} {\ }
    \begin{enumerate}
        \item The length of a greedy solution $G_{K}$ is always $K$. 
        \item A \textit{greedy scheme} is the one that adds a symbol to the existing string at each time $k \text{ for }k \in \{1,\ldots, K-1\}$, so that the resulting string produces the highest value of $f$ without regard to the future times. In other words, a greedy solution is generated by a greedy scheme. 
        \item The $\argmax$ above could be nonunique, in which case there would be multiple greedy solutions. 
    \end{enumerate}
\end{remark}

\begin{definition}
    We define $\beta:= f(G_{K})/f(O_{K})$ as the \textit{performance bound} of the greedy solution. 
\end{definition}

\begin{remark}
    $f(G_{K})$ can be computed exactly, but $f(O_{K})$ is computational intractable. Finding a valid lower bound of $\beta$ is equivalent to finding an upper bound of $f(O_{K})$.
\end{remark}


\section{Main Results}

In this section, we present our bounding results. We first introduce the assumptions that our bounding results rely on. Then we generalize the \textit{greedy curvature} bound of Conforti and Cornu\'{e}jols \cite{conforti1984submodular} from set optimization problems to string optimization problems. Then, we establish a stronger bounding result that relies on fewer assumptions and has a larger bound value.  

Given any string $S = s_{1}s_{2}\cdots s_{|S|} \in \mathbb{S}_{K}$ and $i,j \in \{1,2,\ldots,|S|\}$, denote $S_{i:j} := s_{i}\cdots s_{j}$ if $i\leq j$ and $S_{i:j} = \varnothing$ if $i > j$. For simplicity, we use the abbreviation $S_{k} := S_{1:k} \; (k \in \{0,\ldots,|S|\})$. Note that by definition, $S_{0} = \varnothing$. 

Let $S \in \mathbb{T}$. Define the $\Delta$ notation as follows: for each $k \in \{1,\ldots, |S|\}$, $\Delta(S_{k}) := f(S_{k}) - f(S_{k-1})$, called the \textit{$k^{\text{th}}$ increment} of $f(S)$. In \cite{conforti1984submodular}, $\Delta(S_{k})$ is called a \textit{discrete derivative}. 

For any $S \in \mathbb{T}$, let $\mathbb{S}(S):= \{s\in \mathbb{S}: Ss \in \mathbb{T}\}$. The set $\mathbb{S}(S)$ contains the actions that are feasible with respect to $\mathbb{T}$ following string $S$. The set $\mathbb{S}(G_{k-1}) = \{s \in \mathbb{S}: G_{k-1}s \in \mathbb{T}\}$ is frequently referred to in our derivations. When $k=1$, $\mathbb{S}(G_{k-1}) = \mathbb{S}(\varnothing)$. For $k>1$, $\mathbb{S}(G_{k-1}) = \{s \in \mathbb{S}: G_{k-1}s \in \mathbb{T}\}$ is nonemepty because $g_{k} \in \mathbb{S}(G_{k-1})$ by definition.   


\noindent\textbf{Assumptions:} We introduce the following three key assumptions regarding feasibility along the greedy sequence and the diminishing return property along the optimal and greedy paths:
$$
\begin{aligned}
    \mathbf{A_1:} & \text{ For each } k \in \{1,\ldots,K\}, o_{k} \in \mathbb{S}(G_{k-1}). \\
    \mathbf{A_2:} & \text{ For each } k \in \{1,\ldots, K\}, \Delta(O_{k}) \leq f(o_{k}). \\
    \mathbf{A_3:} & \text{ For each } k \in \{1,\ldots,K\}, \Delta(G_{k-1}s) >0 \text{ for all} \\
    & \; s\in \mathbb{S}(G_{k-1}).
\end{aligned}
$$


Our bounding results (Theorem 1-3) rely on a subset or all of the the above three assumptions. 
We note that $\mathbf{A_2}$ is a weaker conditions than string submodularity, they only involve diminishing return along the optimal sequence instead of any sequence in $\mathbb{T}$. Any string submodular function satisfies $\mathbf{A_2}$, as stated in the following lemma. 

\begin{lemma}
    If $f$ is string submodular, then $f$ satisfies $\mathbf{A_2}$.
\end{lemma}

\begin{proof}
    By condition (2) of string submodularity in Definition 4, we can set $A = \varnothing \text{ and } B = O_{k-1} \text{ for } k\in \{1,\ldots, K\}$. Then $Ao_{k} = o_{k}$ and $Bo_{k} = O_{k}$ are both feasible in the domain of $f$ and $f(O_{k})-f(O_{k-1}) = \Delta(O_{k}) \leq f(o_{k}) \text{ for } k \in \{1,\ldots,K\}$.
\end{proof}

We now establish our first bounding result. We define the greedy curvature for a string function as 
\begin{equation}
\label{greedy_curvature_conforti}
\begin{aligned}
    \alpha_{G} & := \max_{k\in \{2,\ldots,K\}} \max_{ s \in \{s | \Delta(G_{k-1}s) > 0, s \in \mathbb{S}(G_{k-1}) \}} 
    \frac{f(s)}{\Delta(G_{k-1}s)}.
\end{aligned}
\end{equation}

\begin{remark}
\thlabel{alpha_G_1}
    Based on assumptions $\mathbf{A_2}$ and $\mathbf{A_3}$, $\alpha_{G} \geq 1$ always holds. This can be proved by contradiction. Assume $\alpha_{G} < 1$. By the definition of $\alpha_{G}$ and assumptions $\mathbf{A_2}$ and $\mathbf{A_3}$, we have: 
    \begin{equation}
    \label{alphaG_contradict}
    \begin{aligned}
        & \alpha_{G}\Delta(G_{k-1}g_{k}) = \alpha_{G}\Delta(G_{k}) \geq f(g_{k}) \geq f(o_{k}) \geq \Delta(O_{K}) \\
        & \text{for all } k\in \{2,\ldots,K \}.
    \end{aligned}
    \end{equation}
    Adding $f(o_1) \leq f(g_1)$ and summing over $\alpha_{G}\Delta(G_{k}) \geq \Delta(O_{K})$ for all $k \in \{2,\ldots, K\}$ yield: 
    \begin{equation}
    \begin{aligned}
          f(O_K) & = f(o_1) + \sum_{k=2}^{K}\Delta(O_{K}) \leq f(g_1) + \sum_{k=2}^{K} \alpha_{G}\Delta(G_{k}) \\ 
          & < f(g_1) + \sum_{k=2}^{K} \Delta(G_k) = f(G_K),
    \end{aligned}
    \end{equation}
    which brings contradiction with $f(O_K) < f(G_K)$. 
\end{remark}
~\\
\indent Then, we have the following theorem.

\begin{theorem}
\thlabel{GenConforti}
    Assuming $\mathbf{A_1}, \mathbf{A_2} \text{ and }\mathbf{A_3}$, $f(G_K)/f(O_K) \geq \beta_{1}$, where 
    \begin{equation*}
        \beta_{1} = \frac{1}{K} + \frac{1}{\alpha_{G}}*\frac{K-1}{K}.
    \end{equation*}
\end{theorem}

\begin{proof} 
The definition of $\alpha_{G}$ gives us: 
$$
\begin{aligned}
    \alpha_{G} \geq \frac{f(s)}{\Delta(G_{k-1}s)}
\end{aligned}
$$   
for every $k \in \{2,\ldots, K\}$ and $s \in \{s | \Delta(G_{k-1}s) > 0, s \in \mathbb{S}(G_{k-1}) \}.$

By assumption $\mathbf{A_1}$, 
\begin{equation}
\label{beta_greedy_1}
\begin{aligned}
f(o_{k}) \leq \alpha_{G}\Delta(G_{k-1}o_{k})  \text{ for } k \in \{2,\ldots, K\}.
\end{aligned}
\end{equation}

\noindent Moreover, by assumption $\mathbf{A_2}$,
\begin{equation}
\label{beta_greedy_2}
f(O_{K}) = \sum_{k=1}^{K}\Delta(O_{k}) \leq \sum_{k=1}^{K} f(o_{k}) \stackrel{(a)}{\leq} Kf(g_{1}), 
\end{equation}
in which $(a)$ holds since $f(g_1)$, by the definition of greedy solution, dominates each term in $\sum_{k=1}^{K} f(o_{k})$. 

Then, 
\begin{equation}
\label{beta_greedy_3}
    \begin{aligned}
        f(O_{K}) & = \sum_{k=1}^{K} (f(O_{k})-f(O_{k-1})) = \sum_{k=1}^{K} \Delta(O_{k}) \\ 
        & \leq \sum_{k=1}^{K} f(o_{k}) \stackrel{(b)}{\leq} f(g_{1}) + \sum_{k=2}^{K} \alpha_{G}\Delta(G_{k-1}o_{k}) \\
        & \stackrel{(c)}{\leq} f(g_{1}) + \sum_{k=2}^{K} \alpha_{G}\Delta(G_{k}) \\
        & = \alpha_{G}(f(g_1) + \sum_{k=2}^{K}\Delta(G_{K})) + (1-\alpha_{G})f(g_{1}) \\ 
        & = \alpha_{G} f(G_{K}) + (1-\alpha_{G})f(g_{1}),
    \end{aligned}
\end{equation}
where $(b)$ is by the relation in inequality $(\ref{beta_greedy_1})$, and $(c)$ is by the definition of $g_{k}$. 

Dividing  $\alpha_{G} f(O_{K})$ on both sides of $(\ref{beta_greedy_3})$ yields:
\begin{equation}
\label{beta_greedy_4}
    \begin{aligned}
        \frac{f(G_{K})}{f(O_{K})} & \geq \frac{1}{\alpha_{G}} -\frac{1-\alpha_{G}}{\alpha_{G}} * \frac{f(g_{1})}{f(O_{K})} \stackrel{(d)}{\geq} \frac{1}{\alpha_{G}} - \frac{1-\alpha_{G}}{\alpha_{G}} * \frac{1}{K} \\
        \frac{f(G_{K})}{f(O_{K})} & \geq \frac{1}{K} + \frac{1}{\alpha_{G}}*\frac{K-1}{K},
    \end{aligned}
\end{equation}
where $(d)$ is due to inequality $(\ref{beta_greedy_2})$ and the fact that $\alpha_{G} \geq 1$ from \thref{alpha_G_1}.
\end{proof}

We note that the above proof technique is similar to the proof technique used in Theorem 3.1 in \cite{conforti1984submodular} by Conforti and Cornu\'{e}jols, in deriving their greedy curvature bound for set optimization problems. However, Theorem~1 is more general in that it applies to bounding greedy solutions in string optimization problems, which subsume set optimization problems, and it does not rely on having submodularity. The weaker assumptions $\mathbf{A_2}$ suffices. For a set submodular function, the value of the bound in Theorem~1 coincides with the value of the bound of Conforti and Cornu\'{e}jols. The result in Theorem~1 also generalizes our prior result for bounding string optimizations \cite{van2023improved}. 



We now establish a stronger result. 

\begin{theorem}
\thlabel{topK}
    Assuming $\mathbf{A_1}$ and $\mathbf{A_2}$, $f(G_{K}) / f(O_{K}) \geq \beta_{2}$, where 
    \begin{equation}
        \beta_{2} = \frac{f(G_{K})}{\sum_{k=1}^{K} \max_{s\in \mathbb{S}(G_{k-1})} f(s) }. 
    \end{equation}
\end{theorem}

\begin{proof}
$$
\begin{aligned}
    f(O_{K}) = \sum_{k=1}^{K} (f(O_{k})-f(O_{k-1})) = \sum_{k=1}^{K} \Delta(O_{k}) 
\end{aligned}
$$
By $\mathbf{A_2}$, $\Delta(O_{k}) \leq f(o_{k})$ for $k = 1,\ldots,K$.  Therefore,
\begin{equation*}
    \sum_{k=1}^{K}\Delta(O_{k}) \stackrel{(a)}{\leq} \sum_{k=1}^{K} f(o_{k}) \stackrel{(b)}{\leq} \sum_{k=1}^{K}\max_{s\in \mathbb{S}(G_{k-1})}f(s), 
\end{equation*}

\noindent where $(a)$ is by $\mathbf{A_{2}}$ and $(b)$ is by $\mathbf{A_1}$ and the definition of the $\max$ operator. 
\end{proof}

Theorem 2 essentially shows  
that the sum of the objective values under the first $K$ greedy and feasible actions serves a valid upper bound for $f(O_{K})$. This upper bound is easy to compute along the greedy trajectory and is obtained under weaker assumptions than those in Theorem~1.  

We now show that the bound in Theorem~2 is better than the bound in Theorem~1.  

\begin{theorem}
\thlabel{betterThm}
    Assuming $\mathbf{A_1}, \mathbf{A_2} \text{ and } \mathbf{A_3}$, $\beta_{2} \geq \beta_{1}$.
\end{theorem}

\begin{proof}
    Let $s_{k}$ denote the maximizer of $f(s)$ for $s\in \mathbb{S}(G_{k-1})$ where $k\in \{1,\ldots, K\}$. Then $ \sum_{k=1}^{K}\max_{s\in \mathbb{S}(G_{k-1})}f(s) = \sum_{k=1}^{K}f(s_k) \leq Kf(g_1)$ by the definition of $g_1$. The definition of $\alpha_{G}$ further gives us: 
    \begin{equation}
    \label{better_1}
    \begin{aligned}
        \sum_{k=1}^{K}f(s_{k}) & \leq f(g_1) + \sum_{k=2}^{K}\alpha_{G}\Delta(G_{k-1}s_{k}) \\
        & \leq f(g_1) + \sum_{k=2}^{K}\alpha_{G}\Delta(G_{k}) \\
        & = \alpha_{G}f(G_K) + (1-\alpha_{G})f(g_1).
    \end{aligned}
    \end{equation}
    

    Dividing $\alpha_{G}\sum_{k=1}^{K}f(s_{k})$ on both sides of inequality \eqref{better_1} yields:
    \begin{equation}
    \begin{aligned}
        \beta_{2} = \frac{f(G_{K})}{\sum_{k=1}^{K}f(s_{k})} & \geq \frac{1}{\alpha_{G}} - \frac{1-\alpha_{G}}{\alpha_{G}} * \frac{f(g_1)}{\sum_{k=1}^{K}f(s_{k})} \\
        & \stackrel{(a)}{\geq} \frac{1}{\alpha_{G}} - \frac{1-\alpha_{G}}{\alpha_{G}} * \frac{1}{K} \\
        &  = \frac{1}{K} + \frac{1}{\alpha_{G}} * \frac{K-1}{K} = \beta_{1}, 
    \end{aligned}
    \end{equation}
    where inequality (a) is due to the fact that $\sum_{k=1}^{K}f(s_k) \leq Kf(g_1)$ and $\alpha_{G} \geq 1$ from \thref{alpha_G_1}.
\end{proof}

\begin{remark}
    Recall the improved greedy curvature $\alpha_{k}$ proposed in \cite{van2023improved}. For $k \in \{2,\ldots,K\}$,
    \begin{equation*}
    \alpha_{k}:= \max_{s \in \{s | \Delta(G_{k-1}s) > 0, s \in \mathbb{S}(G_{k-1}) \}} \frac{f(s)}{\Delta(G_{k-1}s)}.
    \end{equation*}
    By the definition of $\alpha_{k}$, $f(s_k) \leq \alpha_{k}\Delta(G_{k-1}s_k) \leq \alpha_{k}\Delta(G_{k})$ for $k\in \{2,\ldots,K\}$. Hence, 
    \begin{equation}
    \label{cdc23}
        \sum_{k=1}^{K}f(s_{k}) \leq f(g_1) + \sum_{k=2}^{K}\alpha_{k}\Delta(G_k),
    \end{equation}
    The right hand side of \eqref{cdc23} is the upper bound of $f(O_K)$ in \cite{van2023improved}, which is proven to give a better performance bound than $\beta_{1}$. Therefore, $\beta_{2}$ is also an improved performance bound than that proposed in \cite{van2023improved}.
\end{remark}

\section{APPLICATIONS}

\subsection{Task Scheduling}
As a canonical optimization problem in operations research, task scheduling was studied in \cite{streeter2008online} and further investigated in \cite{zhang2015string} and \cite{liu2018performance}. In task scheduling, we aim to assign agents at multiple stages throughout the task to maximize the probability of successful completion. The detailed mathematical formulation is as follows. 

Assume a task is comprised of $K$ stages, and an agent needs to be assigned at stage $k \; (1\leq k \leq K)$ to accomplish the task. We have a pool of $N$ agents, denoted by $M_{1},\ldots,M_{N}$, and no agent is allowed to be repeatedly selected. For each agent $M_{i} \;(1\leq i \leq N)$, the probability of accomplishing the task at each stage $k \;(1\leq k \leq K)$ is denoted by $p(m_{i}^{k})$. Therefore, the probability of accomplishment after selecting $M_{i_{1}}, M_{i_{2}},\ldots,M_{i_{K}} $ in the end is: 
\begin{equation}
\label{task_scheduling}
    f(M_{i_{1}}M_{i_{2}}\cdots M_{i_{K}}) := 1-\prod_{k=1}^{K} (1-p(m_{i_{k}}^{k})), 
\end{equation}
where the subscript $i_{*}$ in $M$ and $m$ represents 
the selected agent number. Note that the probability of accomplishing the task for each agent is stage-variant. $f(M_{i_{1}}M_{i_{2}}\cdots M_{i_{K}})$ in $(\ref{task_scheduling})$ is henceforth a string function and order dependent. Streeter and Golovin proved in \cite{streeter2008online} that the greedy schemes yields a $\beta_{0} = (1 - e^{-1})$ performance bound when $f$ bears submodularity and monotonicity.  

Consider the following example with $K = 3$, $N = 5$, and the probabilities of successful completion in Table. \ref{Tash Scheduling Table}. The optimal sequence $(O_{3})$ and greedy sequence $(G_{3})$ are equal in this example with $G_{3} = O_{3} = M_{1}M_{2}M_{3}$. We can verify that assumption $\mathbf{A_{1}}$ is satisfied, and $f$ is string submodular. The three performance bounds are as follows:
\begin{equation*}
\begin{aligned}
    &\beta_{2} =  \frac{f(M_{1}M_{2}M_{3})}{f(M_{1})+ f(M_{2}) + f(M_{3})} =  0.7816 ;\\
    & \beta_{1} = \frac{1}{K} + \frac{1}{\alpha_{G}}*\frac{K-1}{K} = 0.5893; \\
    & \beta_{0} = 1- e^{-1} = 0.6321.
\end{aligned}
\end{equation*}

We observe that the $\beta_{2}$ bound outperforms the other two bounds. This example shows the curvature bounds are not necessarily better than the original $(1-e^{-1})$ bound. On some problems they are better, in others they are not. Therefore, we need to choose the largest of the three bounds. But this is easy, because all three bounds are easily compuatble. In the next example, we show a plot that compares the three bounds under different parameters of another optimization problem.

\begin{table}[h!]
    \centering
    \begin{tabular}{||c c c c||} 
     \hline
      & \thead{Stage 1} & \thead{Stage 2} & \thead{Stage 3}\\ [0.5ex] 
     \hline\hline
     $M_1$ &  \thead{0.2} & \thead{0.16} & \thead{0.14} \\ 
     $M_2$ & \thead{0.18} & \thead{0.16} & \thead{0.14} \\ 
     $M_3$ & \thead{0.16} & \thead{0.14} & \thead{0.14} \\
     $M_4$ & \thead{0.14} & \thead{0.12} & \thead{0.10} \\
     $M_5$ & \thead{0.12} & \thead{0.1} & \thead{0.08} \\ [1ex] 
     \hline
    \end{tabular}
    \caption{Probability of successful completion for each agent at each stage.}
    \label{Tash Scheduling Table}
\end{table}

\subsection{Multi-agent Sensor Coverage}

The multi-agent sensor coverage problem was originally studied in \cite{zhong2011distributed} and further analyzed in \cite{sun2019exploiting} and \cite{welikala2022new}. In a given mission space, we need to find a placement of a set of homogeneous sensors to maximize the probability of detecting randomly occurring events. This problem can be formulated as a set optimization problem, which is a special case in the string setting. In this subsection, we demonstrate our theoretical results by applying them to a discrete version of the coverage problem. Our simplified version can be easily generalized to more complicated settings within the same theoretical framework. 

The mission space $\Omega \subset \mathbb{R}^{2}$ is modeled as a non-self-intersecting polygon where $K$ homogeneous sensors will be placed to detect a randomly occurring event in $\Omega$. For simplicity of calculation, we assume both the sensors and the random event can only be placed and occur at lattice points. We denote the feasible space for sensor placement and event occurrence as $\Omega^{F}$. Our goal is to maximize the overall likelihood of successful detection in the mission space, as illustrated in Fig.~\ref{sensors}. 

\begin{figure}[!ht]
    \centering
    \includegraphics[width=0.98\columnwidth]{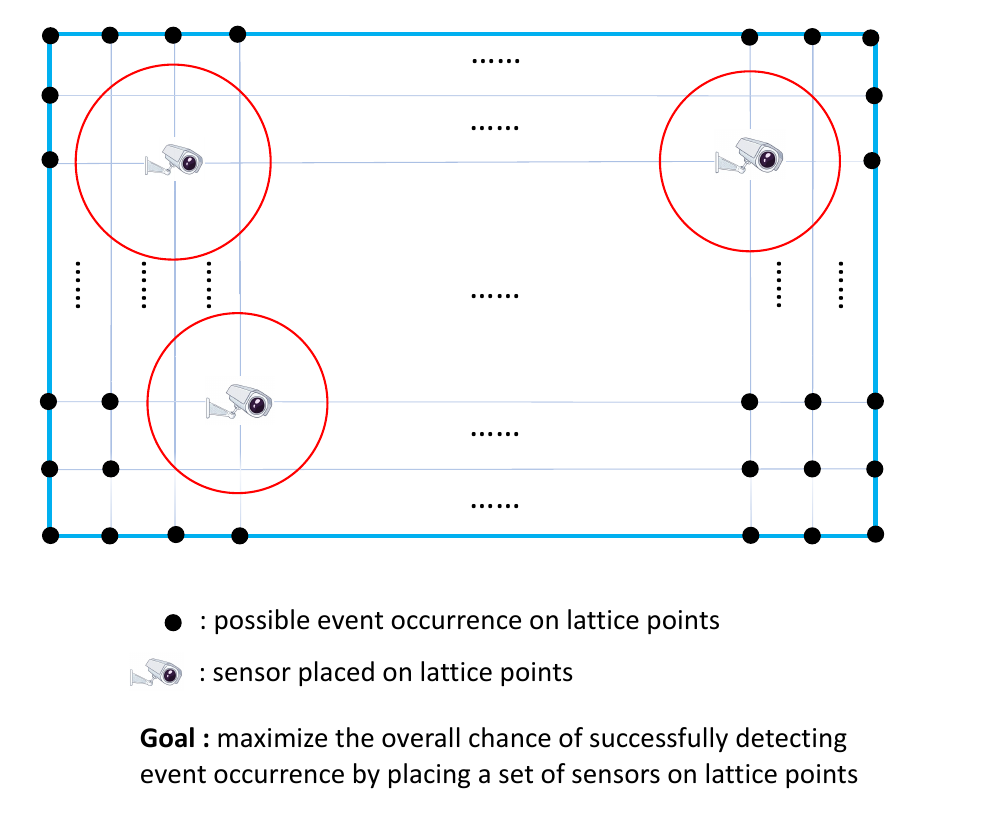}
    \caption{Sensor Coverage in a Mission Space} 
    \label{sensors}
\end{figure}

The likelihood of event occurrence over $\Omega^{F}$ is given by an event mass function $R: \Omega^{F} \xrightarrow{} \mathbb{R}_{\geq 0}$, and we assume that $\sum_{\mathbf{x} \in \Omega^{F}} R(\mathbf{x}) < \infty$. $R(\mathbf{x})$ may reflect a particular distribution if some prior information is available. Otherwise, $R(\mathbf{x}) = 1$ when no prior information is obtained. The locations of all the sensors are represented as $\mathbf{s} = (\mathbf{s}_1,\mathbf{s}_2,\ldots,\mathbf{s}_K) \in (\Omega^{F})^{K}$, where $\mathbf{s}_{i} \; (1\leq i \leq K)$ are coordinates of the placed sensors. Each sensor placed at $\mathbf{s_{i}}$ can detect any occurring event at location $\mathbf{x}$ with probability $p(\mathbf{x},\mathbf{s}_i) = e^{-\lambda \|\mathbf{x}-\mathbf{s}_i\|}$, where $\lambda$ is the decay rate characterizing how quick the sensing capability decays along the distance. 

Assuming all the sensors are working independently, the probability of detecting an occurring event at location $\mathbf{x} \in \Omega^{F}$ after placing $K$ homogeneous sensors at locations $\mathbf{s}$ is $p(\mathbf{x},\mathbf{s}) = 1-\prod_{i=1}^{K}\left( 1-p(\mathbf{x},\mathbf{s}_i) \right)$. In order to consider the whole feasible space for event occurrence, we need to incorporate the event mass function. Our objective function becomes $H(\mathbf{s}) = \sum_{\mathbf{x} \in \Omega^{F}} R(\mathbf{x})p(\mathbf{x},\mathbf{s})$. The goal is to maximize $H(\mathbf{s})$:
\begin{equation}
\begin{aligned}
\label{obj_fun_sensor}
    & \text{maximize } H(\mathbf{s}) \\
    & \text{subject to } \mathbf{s} \in (\Omega^{F})^{K}.
\end{aligned}
\end{equation}


If $n$ lattice points in $\Omega^{F}$ are feasible for sensor placement, we need to select $K$ out of $n$ locations with its complexity being $n!/\left( K!(n-K)! \right)$. This becomes a set optimization problem, and exhaustive search is computationally intractable when $n$ is large. Therefore, greedy algorithm is an approach for an approximate solution in polynomial time. 
It was proved that the continuous version of $H(\mathbf{s})$ is submodular in \cite{sun2019exploiting}, and it is not difficult to verify that its discrete version is also submodular.



In our experiment, we consider a rectangular mission space $\Omega$ of size $50 \times 40$. A set of $K$ homogeneous sensors are waiting to be deployed on those integer coordinates within $\Omega$, denoted by $\Omega^{F}$. For a point $\mathbf{p} = (s_{x},s_{y})$, the event mass function  is given by $R(\mathbf{p}) = \left(s_{x}+s_{y} \right) / \left( s_{x\text{ max}}+y_{y\text{ max}} \right)$, where $s_{x\text{ max}} = 50$ and $s_{y\text{ max}} = 40$ are the largest values of the $x$ and $y$ coordinates in the mission space. This event mass function implies that the random events are more likely to occur around the top right corner of the rectangular mission space.  

A comparison of the performance bounds under different decay rate with $K = 5$ is shown in the upper figure of Fig.~\ref{sensors_bound}. Small decay rates imply good sensing capability and strong submodularity, under which the greedy scheme does not produce nearly optimal objective function value. In the upper figure of Fig. \ref{sensors_bound}, the $\beta_{2}$ bound (red line) always exceeds the $\beta_{1}$ bound (blue line), illustrating \thref{betterThm}.  In addition, we can observe instances in which the $\beta_{2}$ bound is larger than the $\beta_0=(1-e^{-1})$, while the $\beta_1$ bound is below this value. 

A comparison of performance bounds under different number of placed sensors with $\lambda = 1$ is shown in the lower figure of Fig.~\ref{sensors_bound}. We can still observe that the $\beta_{2}$ bound (red line) always exceeds the $\beta_{1}$  bound (blue line) with significant advantages. Both $\beta_{1}$ and $\beta_{2}$ bound decrease as the number of placed sensors $K$ increases. 

\begin{figure}[hbt!]
    \centering
    \includegraphics[width=1\columnwidth]{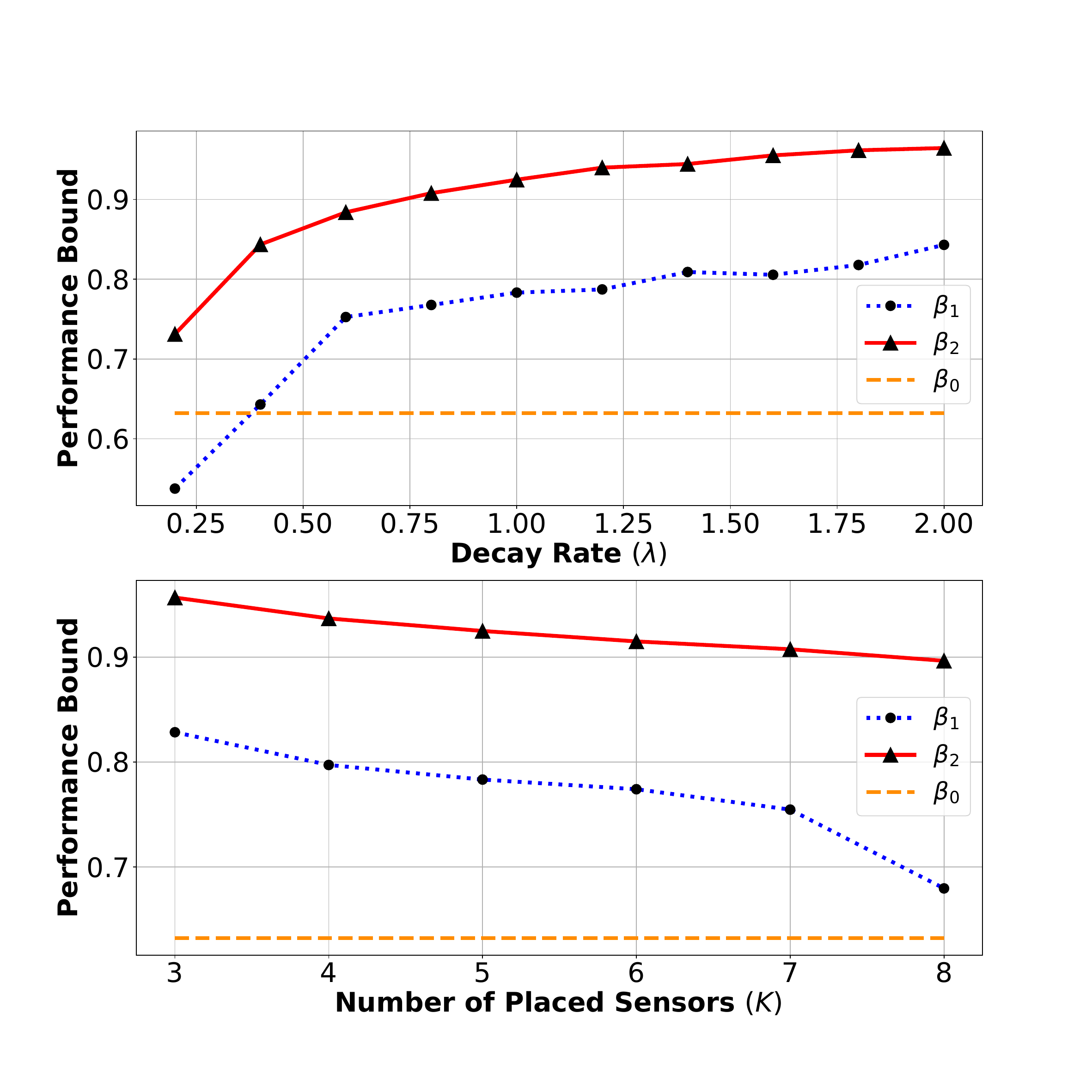}
    \caption{Upper Figure: Performance bound comparison under different decay rates when the number of placed sensors $K=5$; Lower Figure: Performance bound comparison under different number of placed sensors when the decay rate $\lambda = 1$.}
    \label{sensors_bound}
\end{figure}

\section{CONCLUSION}

We derived a computable bound $\beta_{1}$ for greedy solutions to string optimization problems by extending the notion of greedy curvature and the proof technique used by Conforti and Cornu\'{e}jols \cite{conforti1984submodular} from set functions to string functions. However, in deriving our bound we did not use submodularity. Rather we relied on weaker conditions on the string objective function, and hence produced a stronger bounding result. We then derived another computable bound $\beta_{2}$ that relies on even weaker assumptions than those used in deriving our first bound $\beta_{1}$, further strengthening the bounding result. We also showed that this second bound $\beta_{2}$ has a larger value than the first bound $\beta_{2}$. We demonstrated the superiority of our bound on two applications, task scheduling and multi-agent sensor coverage. 



\bibliographystyle{IEEEtran}
\bibliography{root}






\end{document}